\documentclass[10pt]{article}
\usepackage{amsmath}
\usepackage{amsthm}
\usepackage{amssymb}
\usepackage{graphicx}
\usepackage{geometry}
\usepackage[applemac]{inputenc}
\usepackage[all]{xy}
\usepackage{caption}
\usepackage{multicol}

\usepackage[tight]{subfigure}
\usepackage{natbib}

\ifx\pdfoutput\@undefined\usepackage[usenames,dvips]{color}
\else\usepackage[usenames,dvipsnames]{color}
\IfFileExists{pdfcolmk.sty}{\usepackage{pdfcolmk}}{}
\fi
\usepackage[plainpages=false,pdfpagelabels,pagebackref=false,naturalnames=true,hyperindex=true,pdftitle={Heterogeneity extends criticality},pdfauthor={}]{hyperref}
\hypersetup{colorlinks=true,
urlcolor=Cerulean,linkcolor=BrickRed,citecolor=RoyalBlue,
  pdfpagemode=None,
  pdfstartview=FitH}
\usepackage[all]{hypcap}
\usepackage{tikz}

\usepackage[draft,inline,nomargin]{fixme} \fxsetup{theme=color}
\FXRegisterAuthor{cg}{cgg}{\color{red}CGG}

\usepackage[affil-it]{authblk}

\theoremstyle{plain}
\newtheorem{proposition}{Proposition}

\theoremstyle{definition}
\newtheorem{definition}{Definition}
\newtheorem{remark}{Remark}

\newcommand{\ind}{\textnormal{ind}}

\title{Heterogeneity extends criticality}

\author[1,2]{Fernanda S\'anchez-Puig}
\author[1,2,3]{Octavio Zapata}
\author[4]{Omar K. Pineda}
\author[5,6,2,*]{Gerardo I\~niguez}
\author[2,7,8,9*]{Carlos Gershenson}

\affil[1]{\small{Facultad de Ciencias, Universidad Nacional Auton\'{o}ma de M\'{e}xico, 04510 Ciudad de M\'exico, Mexico}}
\affil[2]{\small{Centro de Ciencias de la Complejidad, Universidad Nacional Auton\'{o}ma de M\'{e}xico, 04510 Ciudad de M\'exico, Mexico}}
\affil[3]{\small{Coordinaci\'on de Universidad Abierta, Innovaci\'on Educativa y Educaci\'on a Distancia, Universidad Nacional Auton\'{o}ma de M\'{e}xico, 04510 Ciudad de M\'exico, Mexico}}
\affil[4]{\small{Azure Core Security Services, Microsoft, Redmond, WA 98052, USA}}
\affil[5]{\small{Department of Network and Data Science, Central European University, 1100 Vienna, Austria}}
\affil[6]{\small{Department of Computer Science, Aalto University School of Science, 00076 Aalto, Finland}}
\affil[7]{\small{Departamento de Ciencias de la Computaci\'on, Instituto de Investigaciones en Matem\'aticas Aplicadas y en Sistemas, Universidad Nacional Auton\'{o}ma de M\'{e}xico, 04510 Ciudad de M\'exico, Mexico}}
\affil[8]{\small{Lakeside Labs GmbH, Lakeside Park B04, 9020 Klagenfurt am W\"orthersee, Austria}}
\affil[9]{\small{Santa Fe Institute. 1399 Hyde Park Rd., Santa Fe, NM 87501, USA}}

\affil[*]{\small{Corresponding author email: cgg@unam.mx, iniguezg@ceu.edu}}

\date{}

\begin{document}

\maketitle

\begin{abstract}
Criticality has been proposed as a mechanism for the emergence of complexity, life, and computation, as it exhibits a balance between robustness and adaptability. In classic models of complex systems where structure and dynamics are considered homogeneous, criticality is restricted to phase transitions, leading either to robust (ordered) or adaptive (chaotic) phases in most of the parameter space. Many real-world complex systems, however, are not homogeneous. Some elements change in time faster than others, with slower elements (usually the most relevant) providing robustness, and faster ones being adaptive. Structural patterns of connectivity are also typically heterogeneous, characterized by few elements with many interactions and most elements with only a few. Here we take a few traditionally homogeneous dynamical models and explore their heterogeneous versions, finding evidence that heterogeneity extends criticality. Thus, parameter fine-tuning is not necessary to reach a phase transition and obtain the benefits of (homogeneous) criticality. Simply adding heterogeneity can extend criticality, making the search/evolution of complex systems faster and more reliable. Our results add theoretical support for the ubiquitous presence of heterogeneity in physical, social, and technological systems, as natural selection can exploit heterogeneity to evolve complexity ``for free". In artificial systems and biological design, heterogeneity may also be used to extend the parameter range that allows for criticality.
\end{abstract}

\section{Introduction}

Phase transitions have been studied extensively to describe changes in states of physical matter \citep{stanley1987introduction}, and are typically characterized by symmetry breaking \citep{Anderson1972}. They have also been studied more generally in dynamical systems, such as vehicular traffic \citep{ChowdhuryEtAl2000,helbing2001traffic}. Near phase transitions, critical dynamics are known to occur \citep{Mora2011}. These are also associated with scale invariance and complexity \citep{christensen2005complexity}. There are several examples of criticality in biological systems \citep{RevModPhys.90.031001}, including neural dynamics \citep{doi:10.1098/rsta.2007.2092,Chialvo2010}, genetic regulatory networks \citep{doi:10.1073/pnas.0506771102,Balleza:2008}, and collective motion \citep{Vicsek2012}.

It is often argued that critical dynamics are prevalent or desirable in a broad variety of systems because they offer a balance between robustness and adaptability \citep{Langton1990,Kauffman1993,Hidalgo_2016}. If dynamics are too ordered, then information and functionality can be preserved, but it is difficult to adapt. The opposite occurs with chaotic dynamics: change allows for adaptability, but it also leads to fragility, as small changes percolate through the system. Thus, for phenomena such as life, computation, and complex systems in general, critical dynamics should be favored by evolutionary processes \citep{Gershenson:2010,TorresSosa2012,Roli2018}.

There are different ways in which one can measure criticality, many of which are related to entropies. For example, Fisher information maximizes at phase transitions \citep{wang2011fisher,Prokopenko2011Relating-Fisher}. Still, it rapidly decreases and it is difficult to evaluate how far a system is from criticality. In this work, we use a measure of complexity \citep{Fernandez2013Information-Mea,CxContinuous2016} based on Shannon information that also maximizes at phase transitions, but reduces its value more gradually and is straightforward to calculate compared to Fisher information, as the latter requires to measure the effects of controlled perturbations. There are several definitions and measures of complexity \citep{lloyd2001measures}, but, crucially, the one we use here is highly correlated with criticality.

If criticality is found only near a phase transition, then most of a parameter space would have ``undesirable'' solutions. Thus, how can a search procedure find the right parameters for criticality? Self-organized criticality \citep{BTW1987,Adami1995,hesse2014self,Vidiella2020.11.16.385385} has been proposed as an answer. Although interesting and useful for specific cases, it is not universal and has hidden variables. In general, one can think of different mechanisms that will find or adjust parameters so that criticality is achieved. But, could criticality be more prevalent than previously thought?

In previous work where we have studied rank dynamics in a variety of systems \citep{Cocho2015,Morales2016,10.3389/fphy.2018.00045,Iniguez2022}, we observe that the most relevant elements change more slowly than less relevant elements. We hypothesized that heterogeneous temporality equips systems with robustness and adaptability at the same time. Here we explore the role of heterogeneity in different dynamical systems. We show that different types of heterogeneity extend the parameter region where critical dynamics are observed. Thus, we can say that heterogeneity results in ``criticality for free", reducing the problem of fine-tuning parameters.

\section{Results}

We first present results of a heterogeneous version of the Ising model, where elements have different temperatures. We then explore structural and temporal heterogeneity in random Boolean networks. Afterwards, we abstract the specific dynamics of a system and investigate under which conditions heterogeneity promotes criticality. Finally, we provide a general solution, independent of any measure, using Jensen's inequality.

\subsection{Value heterogeneity: the Ising model}
\label{sec:ising}

We can consider a system of interacting atoms arranged  in a network-like structure (Fig.~\ref{fig:ising_cyclic}).
The state of an atom is defined by its dipole nuclear magnetic moment: a two-valued spin representing the orientation of the magnetic field produced by the atom.
Intuitively, neighboring atoms with the same spin value contribute less to the total energy of the system than atoms with different spin values.
Systems of this kind  evolve preferentially to states with the lowest possible energy. When the temperature of the environment is increased, the system heats, and we can observe a sudden change in a global property of the system, namely loss of magnetization.
A theoretical model of such a system of atoms is the Ising model  \citep{ising1925beitrag,Glauber1963TimeDependent-S}.

\begin{figure}[t]
\begin{center}
\subfigure[]{
   \includegraphics[width=.30\textwidth]{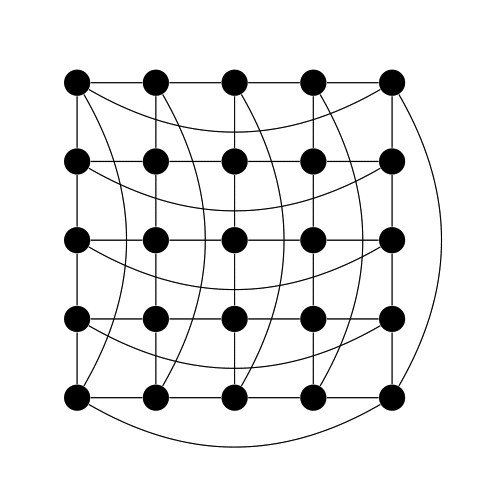}
    \label{fig:ising_cyclic}
}
\subfigure[]{
    \includegraphics[width=.44\textwidth]{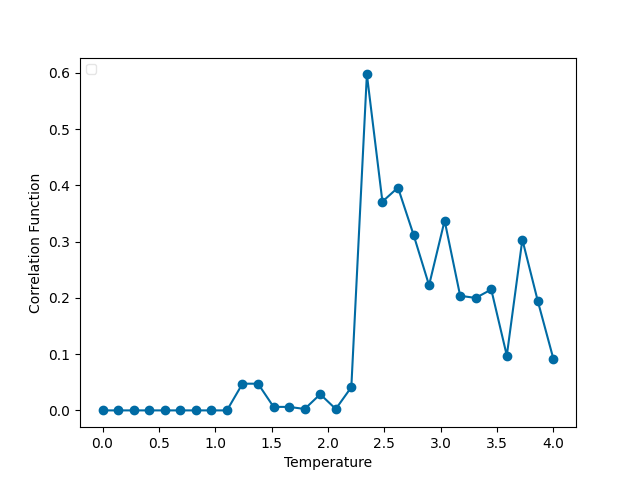}
    \label{fig:ising_correlation}
}
\subfigure[]{
    \includegraphics[width=.44\textwidth]{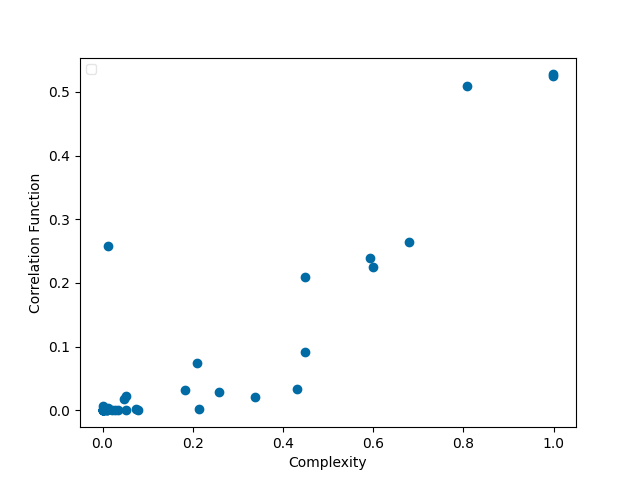}
   \label{fig:Ising_correlation}
    }
\subfigure[]{
    \includegraphics[width=.44\textwidth]{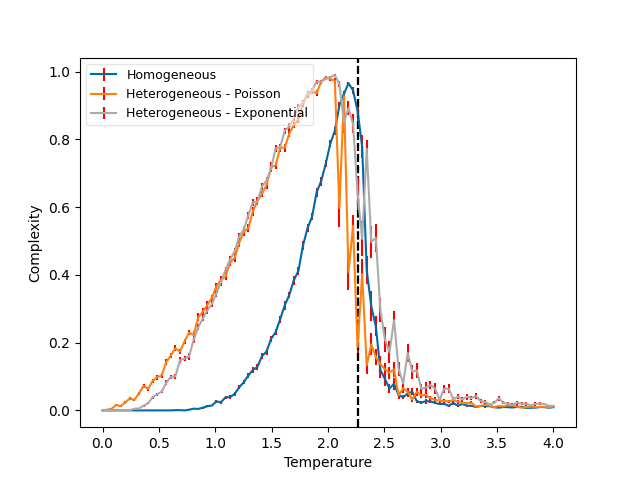}
    \label{fig:Ising_complexity}
    }
\caption{
(A) Two-dimensional Ising model displayed on a square lattice. The graph may be wrapped into a torus, highlighting periodic boundary conditions.
(B) The correlation function decays faster at low and high temperatures than at the critical temperature where the correlation function is maximum.
(C) Correlation as a function of complexity in two-dimensional Ising model illustrates that complexity is a good proxy for criticality.
(D) Average complexity with error bars of the Ising model for different temperatures, considering homogeneous (blue), heterogeneous with exponentially distributed (green), and heterogeneous with Poisson distributed (orange) temperatures. The black dotted vertical line represents the theoretical phase transition at $T\approx2.27$ (in practice smaller due to finite size effects).
}
\label{fig:ising}
\end{center}
\end{figure}

The Ising model is usually homogeneous: all cells have the same temperature, and one explores different properties as the temperature $T$ varies. This is a good assumption when all atoms can be considered to behave in a similar way. However, if we are modeling an Ising-like biological system \citep{Hopfield1982}, then each element might have slightly different properties.
In the proposed heterogeneous case, each cell has a temperature taken from a Poisson distribution with a mean equal to the temperature of the homogeneous case (see Sec. \ref{sec:Ising} for details).

Following \citet{LopezRuiz:1995}, we have proposed a measure of complexity \citep{Fernandez2013Information-Mea} based on Shannon's information \citep{Shannon1948},
\begin{equation}
I = -K \sum_{i=i}^{b} p_{i} \log p_{i},
\label{eq:I}
\end{equation}
where $K$ is a positive constant and $b$ is the length of the alphabet (for all the cases considered in this paper, $b=2$). This measure is equivalent to the Boltzmann-Gibbs entropy. To normalize $I$ to $[0,1]$, we use
\begin{equation}
K = \frac{1}{\log_{2}b}.
\end{equation}
$I$ is maximal when the probabilities are homogeneous, i.e. there is the same probability of observing any symbol along a string. $I$ is minimal when only one symbol is found in a string (so it has a probability of one, and all the rest have a probability of zero). Chaotic dynamics are characterized by a high $I$, while ordered (static) dynamics are characterized by a low $I$. Inspired by \citet{LopezRuiz:1995}, we define complexity $C$ as the balance between ordered and chaotic dynamics,
\begin{equation}
C = 4 \cdot I \cdot (1-I),
\label{eq:C}
\end{equation}
where the constant 4 is added to normalize the measure to $[0,1]$ \citep{10.3389/frobt.2017.00010}.

Figure~\ref{fig:ising_correlation} shows the correlation of the Ising model for varying temperature. This is maximal in the phase transition at $T\approx2.27$, i.e. criticality. Figure~\ref{fig:Ising_correlation} shows that there is a correspondence between the correlation and the complexity measure in Eq.~\ref{eq:C}. Figure~\ref{fig:Ising_complexity} shows results of average complexity $C$ as $T$ increases. Complexity is maximal near the phase transition for the homogeneous case. Heterogeneity shifts the expected maximum complexity (that reflects criticality), but it also expands it, in the sense that the area under the curve is broadened. In other words, critical-like dynamics (one can assume arbitrarily complexity values greater than 0.8, just for comparison) are found for a broader range of $T$ values.

\subsection{Temporal and structural heterogeneity: random Boolean networks}
\label{sec:boolean}

A gene is a part of the genomic sequence that encodes how to produce (synthesise) either a protein or some RNA (gene products).
Gene product synthesis is called gene expression.
Because not all gene products are synthesised at the same time, the regulation of gene expression is constantly taking place within a cell.
In fact, the expression of each gene is regulated (among many things) by the expression of other genes in the genome.
 This gives rise to an interaction structure known as a genetic regulatory network.
Boolean networks are a theoretical model of genetic regulatory networks.
In random Boolean networks (RBNs) \citep{Kauffman1969,Kauffman1993}, traditionally there is homogeneous topology and updating. In this case, critical dynamics are found close to a phase transition between ordered and chaotic phases \citep{DerridaPomeau1986,LuqueSole1997,Wang:2010}.

\begin{figure}[t]
\begin{center}
\begin{subfigure}[]{
    $\xymatrix@C=30pt@R=15pt{
    1\ar@/^/[dd]&&2\ar@/^/[ld]&\\
    &3\ar@/^/[ru]\ar[dr]\ar[dd]\ar@(ul,dl)[]&&4\ar@(ru,rd)[]\ar@/^/[dl]&\\
    5\ar@(ul,dl)[]\ar@/^/[uu]\ar[dr]&&6\ar[uu]\ar@/^/[ru]&\\
    &7\ar[uuul] &   }$
    \label{fig:RBN_structure}
    }
\end{subfigure}
\begin{subfigure}[]{
    \begin{tabular}{c|c|c|c|c|c|c|c}
    $xy$ &$f_1$ &$f_2$&$f_3$ &$f_4$&$f_5$ &$f_6$&$f_7$   \\ \hline
    $00$&$0$&$0$&$0$&$0$&$0$&$1$&$1$ \\ \hline
    $01$&$0$&$1$&$1$&$1$&$0$&$0$&$0$ \\ \hline
    $10$ &$0$&$1$&$1$&$0$&$0$&$0$&$1$  \\ \hline
    $11$ &$1$&$1$&$0$&$1$&$0$&$0$&$0$  
    \end{tabular}
    \label{fig:RBN_table}
    }
\end{subfigure}
\begin{subfigure}[]{
    \includegraphics[width=0.5\textwidth]{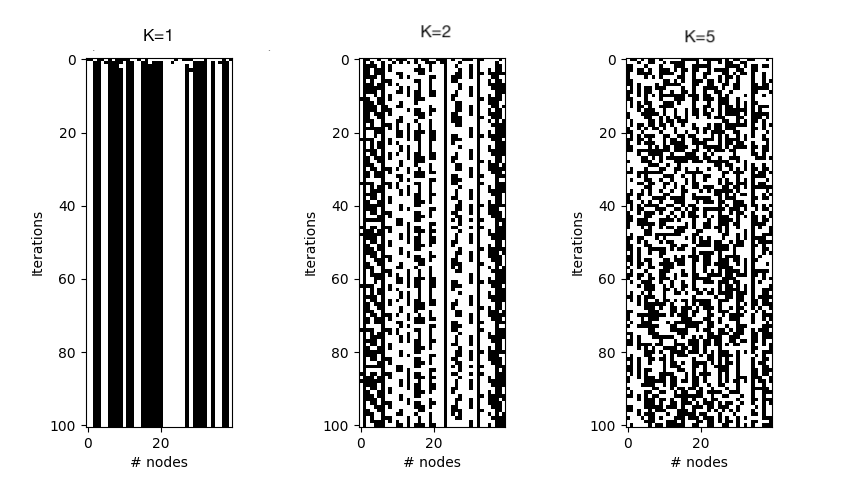}
    \label{fig:crbn-k1}
}
\end{subfigure}
\begin{subfigure}[]{
    \includegraphics[height=0.3\textwidth]{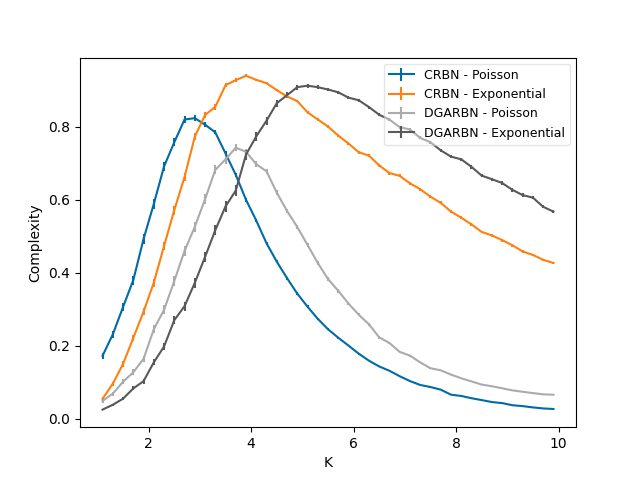}
    \label{fig:cxRBN-plot}
    }
\end{subfigure}
\caption{
(A) Example of a $k$-in regular directed graph with set of nodes $V=\{1,2,\dots,7\}$ ($N=7$) and $K=2$.
(B) Truth table of the functions comprising a Boolean network with $7$ nodes and $K=2$.
(C) Example of three regimes of CRBN and their measures of complexity using 40 nodes ($N=40$) with 100 steps each. (time flows downwards)
For $K=1$, $C=0.0558$.
For $K=2$, $C=0.9951$.
For $K=5$, $C=0.4714$.
(D) Average complexity of RBNs as the average connectivity $K$ is increased. Combinations of ``homogeneous'' structure (Poisson), heterogeneous structure (Exponential), homogeneous temporality (CRBN), and heterogeneous temporality (DGARBN).  $\Delta K$ = 0.2, $N$=100, with 1000 iterations for each $K$.
}
\label{fig:cxRBN}
\end{center}
\end{figure}

Figure~\ref{fig:RBN_structure} shows an example of the topology of a RBN with seven nodes ($N=7$) and two connections (inputs $K$) each. Each node has a lookup table where all possible combinations of their inputs are specified (e.g. Figure~\ref{fig:RBN_structure}). Using an ensemble approach, for each parameter combination, we randomly generate topologies (structure) and lookup tables (function), and then evaluate them in simulations. Depending on different parameters, the dynamics of RBNs can be classified as ordered, critical (near a phase transition), and chaotic. Figure~\ref{fig:crbn-k1} shows example of these dynamics for different $K$ values.

One can have heterogeneous topology in different ways \citep{OosawaSavageau2002,Aldana2003}, as genetic regulatory networks are not homogeneous: few genes affect many genes, and many genes affect few genes. 
Here, we use Poisson and exponential distributions. Strictly speaking, both are heterogeneous, but exponential is more heterogeneous than Poisson, which here we consider as ``homogeneous''. The technical reason for using a Poisson distribution is that it allows us to explore non-integer average connectivity in the network.

We can also have heterogeneous updating schemes \citep{Gershenson2002e}, as it can be argued that not all genes in a network ``march in step'' \citep{HarveyBossomaier1997}. Classical RBNs (CRBNs) have synchronous, homogeneous temporality, while in here we use Deterministic Generalized Asynchronous RBNs (DGARBNs) for heterogeneous temporality. In particular, each node is updated every number of time steps equal to its out-degree, so the more nodes one node affects, the slower it will be updated (see Sec~\ref{sec:RBN} for details).



Fig.~\ref{fig:cxRBN-plot} compares the average complexity $C$ as the average connectivity $K$ is increased. Structural and temporal homogeneity (CRBN-Poisson) has a classical complexity profile, maximizing near the phase transition ($K=2$ for the thermodynamical limit, i.e., $N\rightarrow \infty$).
It can be seen that only structural heterogeneity (CRBN-Exponential) extends criticality more than only temporal heterogeneity (DGARBN-Poisson), that basically shifts the curve to the right. Still, having both structural and temporal heterogeneity (DGARBN-Exponential) extends criticality even more than having only structural heterogeneity.

\subsection{Arbitrary complexity}
\label{sec:cx}


Abstracting the results from the previous subsections, and trying not to depend on any model in particular, we can explore exhaustively the measure of complexity (Eq.~\ref{eq:C}) in homogeneous and heterogeneous settings, to observe when each case yields a higher average complexity. So we simply vary the probability $p_1$ of having ones in a binary string directly as shown in Figure \ref{fig:analitical1}.

In the homogeneous case, we calculate directly the complexity $C$ as a function of $p_1$ using Eq.~\ref{eq:C}, assuming that we are averaging the complexities of several elements with the same $p_1$. For the heterogeneous case, we generate a collection of probabilities with mean $p_1$ and standard deviation of 0.2  (truncating to zero negative values and to one values greater than one), calculate their complexity, and then average it. Heterogeneity achieves higher complexities for roughly $0.25 <p_1 <0.75$. One might wonder why all heterogenous complexities avoid extreme values, even when heterogeneous RBNs can have complexities close to zero and one. This is because of the standard deviation of the distributions from which the means are generated. Smaller standard deviations yield curves closer to the heterogeneous case.


\begin{figure}[t]
\begin{center}
\subfigure[]{
	\includegraphics[scale=0.5]{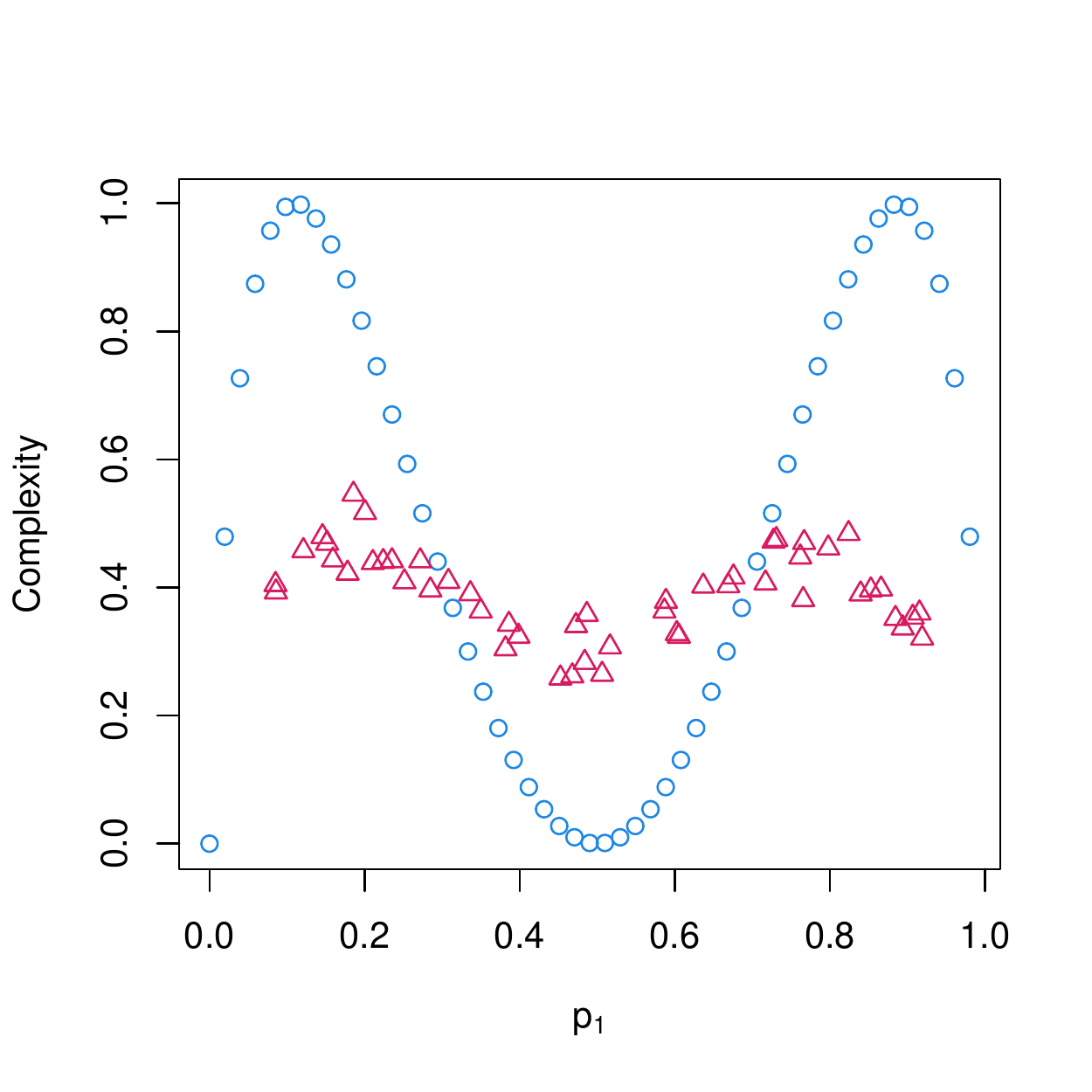}
	\label{fig:analitical1}
}
\subfigure[]{
	\includegraphics[scale=0.65]{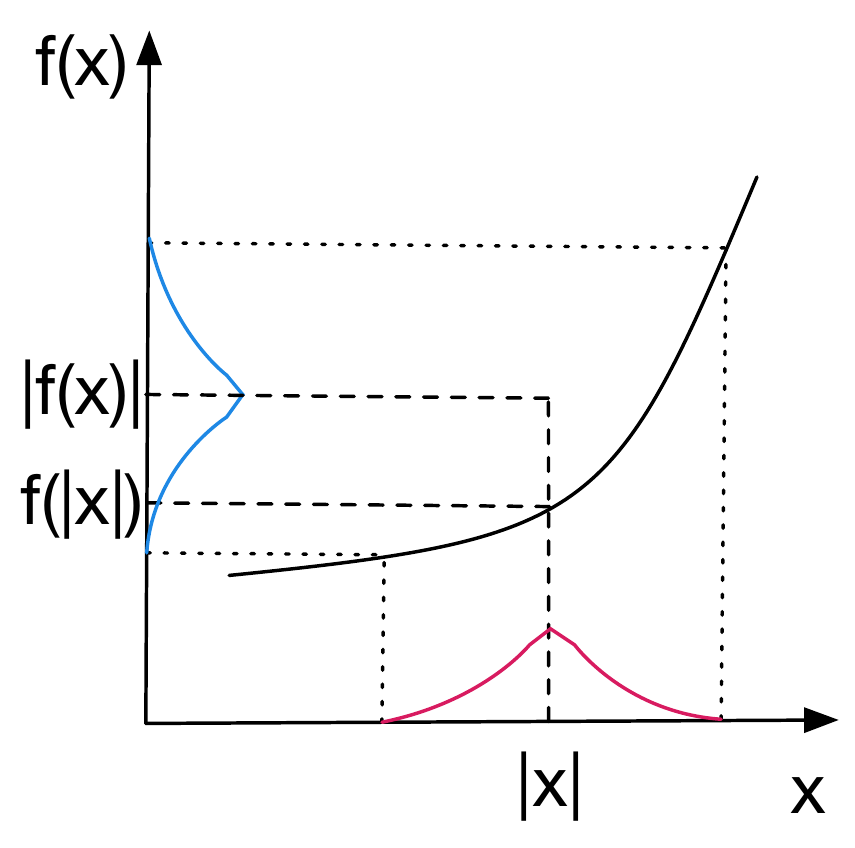}
	\label{fig:Jensen}
}

\caption{A. Average complexity $C$ for collections of strings with average probability of ones $p_1$, in homogeneous (blue circles) and heterogeneous (red triangles) cases. The latter yields higher average complexity in the central region, where the homogeneous complexity is low. 
B. Illustration of Jensen's inequality. The function of the averages $f(|x|)$ of a variable with a distribution with average $|x|$ is lower than the average of the functions $|f(x)|$ for concave functions. The opposite is the case for convex functions.}

\end{center}
\end{figure}

%
%


By assuming that heterogeneity sometimes will be better than homogeneity and vice versa, we can further generalize our results to be independent of any measure or function. If we have homogeneity of a variable $x$, all elements will have the same value for $x$, and thus the mean $|x|$ will be equal to any $x_i$. Thus, the average of any function $|f(x)|$ will be equal to any $f(x_i)$. If we have heterogeneity, then the mean $|x|$ will be given by some distribution of different values of $x$, and similarly for $|f(x)|$.

We can then say that heterogeneity is preferred when the function of the average is greater than the average of the function,
\begin{equation}
f(|x|)>|f(x)|.
\end{equation}

Jensen's inequality \citep{mcshane1937jensen} tells us already that heterogeneity will be ``better'' than homogeneity for concave functions, as illustrated in Figure~\ref{fig:Jensen}. For more complex functions, their concave parts will benefit from heterogeneity and their convex parts will benefit from homogeneity (as it can be seen for $C$ in Figure \ref{fig:analitical1}).

For linear functions, it can be shown that there is no difference between homogeneity and heterogeneity, as $f(|x|)$ will always be equal to $|f(x)|$ (see proof in Section~\ref{sec:linear}). Thus, it can be concluded that the difference between homogeneity and heterogeneity is relevant only for nonlinear functions.





\section{Discussion}

There are several recent examples of heterogeneity offering advantages when compared to homogeneous systems in the literature. For example, in public transportation systems, theory tells us that passengers are served optimally (wait at stations for a minimum time) if headways are equal, i.e., homogeneous. However, equal headways are unstable \citep{GershensonPineda2009,Chen2021}. Still, adaptive heterogeneous headways can deliver supraoptimal performance through self-organization \citep{Gershenson:2011a,10.1371/journal.pone.0190100}, due to the slower-is-faster effect \citep{GershensonHelbing2015}: passengers do wait more time at stations, but once they board a vehicle, on average they will reach faster their destination, as the idling required to maintain equal headways is avoided.

There are other examples where heterogeneity promotes synchronization (see \citet{Zhange2024299118} and references therein). In particular, \citet{Zhange2024299118} shows that random parameter heterogeneity among oscillators can consistently rescue the system from losing synchrony. In related work, \citet{Molnar2021} finds that heterogeneous generators improve stability in power grids. Recently, \citet{math9212769} explored complex networks with heterogeneous nodes, observing that these have a greater robustness as compared to networks with homogeneous nodes. In social networks, \citet{ZHOU202047} finds that heterogeneity of social status may drive the network evolution towards self-optimization. Also, structural heterogeneity has been shown to favor the evolution of cooperation \citep{Santos2006,Santos2008}.

These examples suggest that heterogenous networks improve information processing. With heterogeneity, elements can in principle process information differently, potentially increasing the computing power of a heterogeneous system over an homogeneous one with similar characteristics. This is related to Ashby's law of requisite variety \citep{Ashby1956,Gershenson2015Requisite-Varie}, which states that an active controller should have at least the same variety (number of states) as the controlled. It is straightforward to see with random Boolean networks that temporal heterogeneity increases the variety of the system: the state space (of size $2^N$ for homogeneous temporality) can explode once we have to include the precise periods and phases of all nodes (in heterogeneous temporality), as different combinations of the temporal substates may lead a transition from the same node substate to different node substates. Also in random Boolean networks, higher $K$ implies more possible nets. Even if there are evolutionary pressures for efficiency (smaller networks), if heterogeneity shifts criticality to higher $K$, then it will be easier for an evolutionary search to find critical dynamics in larger spaces.

Shannon's~\citeyearpar{Shannon1948} information, equivalent to Boltzmann-Gibbs entropy, is maximal when the probability of every symbol or state is the same, i.e. homogeneous. Thus, one can measure heterogeneity as an inverse of entropy (one minus the normalized Shannon's information) \citep{Fernandez2013Information-Mea}. It is clear that maximum heterogeneity (as measured here, it would occur when only one symbol or state has a probability of one and all the rest a probability of zero) has its limitations. Thus, we can assume that there will be an ``optimal'' balance between minimum and maximum heterogeneities. The precise balance will probably depend on the system, its context, and may even change in time. If we want heterogeneity to take the dynamics towards criticality (or somewhere else), then the precise ``optimal'' heterogeneity will depend on how far we are from criticality \citep{Gershenson:2010,Pineda2019}. In this sense, a potential relationship with no-free-lunch theorems \citep{wolpert95no,wolpert1997no} seems an interesting area of further research.

When homogeneous systems are analyzed in terms of their symmetries, heterogeneity is a type of symmetry breaking. Still, in converse symmetry breaking \citep{PhysRevLett.117.114101}, only heterogeneity leads to stability, i.e. the system symmetry is broken to preserve the state symmetry. This idea can be used to control the stability of complex systems using heterogeneity \citep{Nicolaou2021}. A further avenue of research is the relationship between heterogeneity and L\'{e}vy flights \citep{Iniguez2022}. L\'{e}vy flights are heterogeneous, since they consist of many short jumps and few large ones. They offer a balance between exploration and exploitation, and seem advantageous for foraging \citep{Ramos-Fernandez2004}, preventing extinctions \citep{Dannemann3794}, and search algorithms \citep{Martinez-Arevalo2020}. Another interesting relationship to study is the one between heterogeneity and non-reciprocal systems \citep{Fruchart2021}.




Network science \citep{albert2002statistical,Newman:2003,Barabasi2016} has demonstrated the relevance of structural heterogeneity. This should be complemented with a systematic exploration of temporal \citep{Barabasi2005} and other types of heterogeneity. For example, it would be interesting to study heterogeneous adaptive \citep{gross2009adaptiveNets} and temporal \citep{Holme2012TemporalNets,holme2015modern} networks, where each node has a different speed for its dynamics. Temporal heterogeneity enables a system to match the requisite variety of their environment at different timescales. If systems can adapt at the scales at which their environments change, then they will better do so if they have a variety of timescales, i.e., heterogeneous temporality. Recently, \cite{Sormunen2022} have shown that adaptive networks have critical manifolds that can be navigated as parameters change. In other words, criticality is not restricted to a single value, but can be associated to a manifold in a multidimensional system.

Further research is required to better understand the role of heterogeneity in the criticality of complex systems. The present work is limited and many open questions remain. We encourage the reader to experiment with a heterogeneous version of their favorite homogeneous complex system model, be it structural, temporal, or other type of heterogeneity. We could learn more from heterogeneous models of collective motion, opinion formation, financial markets, urban growth, and more. This could contribute to a broader understanding of heterogeneity and its relationship with criticality.






\section{Methods}


A \emph{graph} $G$ consists of a set of \emph{vertices} $V$ and a set of \emph{edges} $E$, where an edge is an unordered pair of distinct vertices of $G$.
We write $u\sim v$ to denote that $\{u,v\}$ is an edge and in this case we say that $u$ and $v$ are \emph{adjacent}.
If $H$ is a graph with vertex set $W \subset V$ and edge set $F \subset E$, we say that $H$ is a \emph{subgraph} of $G$.
A graph is said to be \emph{connected} if for every pair of distinct vertices $u$ and $v$, there is a finite sequence of  distinct vertices $a_0,a_1\dots,a_n$ such that $a_0 = u$, $a_n = v$, and $a_{i-1} \sim a_i$ for each $i = 0,1,\dots,n$.
A \emph{connected component} of $G$ is a connected subgraph of $G$.
A graph is said to be finite just in case its vertex set is finite.
A graph is called $d$-\emph{regular} if every vertex is adjacent to exactly  $d\geq 1$ distinct vertices.

A \emph{directed graph} $D$ consists of a set $V$ of elements $a,b,c,\dots$ called the \emph{nodes} of $D$ and a set $A$ of ordered pairs of nodes $(a,b), (b,c), \dots$ called the  \emph{arcs} of $D$.
We use the symbol $ab$ to represent the arc $(a,b)$.
If $ab$ is in the arc set $A$ of $D$, then we say that $a$ is an  \emph{incoming neighbour} (or \emph{in-neighbour}) of $b$, and also that $b$  is a \emph{outgoing neighbour} (or \emph{out-neighbour}) of $a$.
We say that $D$ is \emph{$k$-in regular}  ($k\geq 1$)  if every node has exactly $k$ in-neighbours: for every node $a$ there are distinct nodes $a_1,\dots,a_k$, such that $a_ja\in A$ for $j=1,\dots,k$.
In other words, $D$ is $k$-in regular just in case the set of  in-neighbours of any node has exactly $k$ elements, all distinct, and possibly including itself.
The \emph{out-degree} of a node $a$ is the number of nodes $b$ such that the arc $ab$ is in the arc set of $D$.
Thus the out-degree of $a$ is the number of out-neighbours of $a$.
Similarly, the \emph{in-degree} of a node $a$ is the number of nodes $c$ such that $ca \in A$.
Thus the in-degree of $a$ is the number of in-neighbours of $a$.

\subsection{The Ising model with individual temperatures}
\label{sec:Ising}


It is quite common to study the Ising model on a finite, connected $4$-regular graph where the number of edges is twice the number of vertices.
This graph is usually introduced  as a finite lattice of two-dimensional  points on the surface of a three-dimensional torus.
An example of such a graph with $25$  vertices and $50$ edges is shown in Figure~\ref{fig:ising_cyclic}.

\subsubsection{The Ising model}\label{sec:model}

We start with a finite graph $G = (V, E)$.
We identify the vertex set of $G$ with a system of interacting atoms.
Each vertex $u \in V$ is assigned a  \emph{spin} $\sigma_u$ which can take the value $+1$ or $-1$.
The \emph{energy} of a configuration of spins is
\[
H(\sigma)\ =\  -\sum_{\substack{u,v \in V \\ u\sim v}} \sigma_u\sigma_v.
\]
The energy increases with the number of pairs of adjacent vertices having different spins.
The Ising model is a way to assign probabilities to the system configurations.
The probability of a configuration $\sigma$ is proportional to $\exp(-\beta H(\sigma))$, where $\beta \geq 0$ is a variable inversely proportional to the temperature.

More precisely, the \emph{Ising model} with inverse temperature $\beta$ is the probability measure $\mu$ on the set of configurations $X = \{+1,-1\}^V$ defined by
\[
\mu(\sigma)\ =\ \frac{1}{Z}  \exp(-\beta H(\sigma))
\]
where $Z = Z(G,\beta)$ is a normalizing constant.
This constant can be computed explicitly as
\[
Z(G,\beta)\ =\ \exp(-\beta |E|)\sum_{F\subset E} (\exp(\beta)-1)^{|F|} 2^{k\langle F\rangle}
\]
where $|A|$ denotes the cardinality of a finite set $A$, and $k\langle F\rangle$ the number of connected components of  the (spanning) subgraph $\langle F\rangle = (V, F)$ of $G$.
Then
\[
\lim_{\beta \to 0}\ Z(G,\beta) \ = \ C
\]
where $C = \sum_{F\subset E}2^{k\langle F\rangle}$ and so, for any configuration $\sigma$, we have that
\[
\lim_{\beta \to 0}\ \mu(\sigma) \ =\ \frac{1}{C}.
\]
 As the temperature increases (and hence $\beta \to 0$), $\mu$ converges to the uniform measure over the space of configurations.
When the temperature decreases, $\beta > 0$ increases, and $\mu$ assigns greater probability to configurations that have a large number of pairs of adjacent vertices with the same spin.

\subsubsection{Simulation}
Most simulations of the Ising model use either the Glauber dynamics or the Metropolis algorithm for constructing a Markov chain with stationary measure $\mu$.
Here we only describe the Metropolis chain for the Ising model.

Given two configurations $\sigma,\sigma' \in X$, let $P(\sigma,\sigma')$ denote the probability that the Metropolis chain for the Ising model moves from $\sigma$ to $\sigma'$.
For every $a \in V$, we write $\sigma^a$ to denote the configuration obtained from $\sigma$ by flipping the sign of the value that $\sigma$ assigns to $a$ and leaving all the other spins the same.
 In other words, $\sigma^a \in X$ is the unique configuration which agrees everywhere with $\sigma$ except for the spin assigned to vertex $a$: for every $u \in V$,  $\sigma_{u}^{a} = \sigma_u$ if $u \neq a$ and $\sigma^{a}_{u} = - \sigma_u$ if $u = a$.
 We let the transition probabilities to be positive $P(\sigma, \sigma') > 0$ just in case $\sigma' = \sigma$ or $\sigma' = \sigma^a$ for some $a\in V$.
In the latter case, the Metropolis chain moves from $\sigma$ to $\sigma^a$ with probability
 \[
P(\sigma,\sigma^a)\ =\ \frac{1}{|V|} \left(1\ \wedge\ \frac{\mu(\sigma^a)}{\mu(\sigma)} \right)
\]
where $x \wedge y$ denotes the minimum of the quantities $x$ and $y$.
The probability that the chain stays at the same configuration $\sigma$ is then
 \[
P(\sigma,\sigma)\ =\ 1 - \sum_{a\in V} P(\sigma,\sigma^a).
\]

A key property about these transition probabilities is that they only depend on the ratios $\mu(\sigma^a)/\mu(\sigma)$. Therefore, to simulate the Metropolis chain it is not necessary to compute the normalizing constant $Z$ of the Ising measure $\mu$.

To summarize, we have constructed a transition matrix $P$ that defines a reversible Markov chain with stationary measure $\mu$.

\begin{proposition}
The Metropolis chain for the Ising model has stationary measure $\mu$.
\end{proposition}
\begin{proof}
It is sufficient to prove that the probability measure $\mu$ and the transition matrix $P$ satisfy the detailed balance equations
\begin{equation}\label{eq:1}
\mu(\sigma) P(\sigma,\sigma') = \mu(\sigma')P(\sigma',\sigma)
\end{equation}
for all $\sigma \neq \sigma'$.
To show this, it suffices to verify that  the equation~\eqref{eq:1} holds when $\sigma' = \sigma^a$ for some $a\in V$.
After cancellation of $1/|V|$ and distributing $\mu(\sigma)$ and $\mu(\sigma^a)$ accordingly,  it suffices to check
\[
\mu(\sigma)\ \wedge\ \mu(\sigma)\frac{\mu(\sigma^a)}{\mu(\sigma)}  \ = \ \mu(\sigma^a)\ \wedge\ \mu(\sigma^a)\frac{\mu(\sigma)}{\mu(\sigma^a)}
\]
or equivalently
\[
\mu(\sigma)\ \wedge\ \mu(\sigma^a)  \ = \ \mu(\sigma^a)\ \wedge\ \mu(\sigma)
\]
which is obvious.
\end{proof}

\subsubsection{Individual temperatures}
 In the previous section, we described how to construct a transition matrix $P$ that defines a reversible Markov chain with stationary measure $\mu$.
 Starting at  a configuration $\sigma$, the probability that the chain moves to a new configuration $\sigma^a$ for any $a \in V$, is given by
 \begin{align*}
P(\sigma,\sigma^a)\ &=\ \frac{1}{|V|} \left(1\ \wedge\ \frac{\mu(\sigma^a)}{\mu(\sigma)} \right)\\
&=\ \frac{1}{|V|} \left(1\ \wedge\ \frac{\exp(-\beta H(\sigma^a))}{\exp(-\beta H(\sigma))} \right)\\
&=\ \frac{1}{|V|} \left(1\ \wedge\ \exp(-\beta\Delta H_a(\sigma) ) \right)
\end{align*}
where
\begin{align*}
 \Delta H_{a}(\sigma)\ &= \ H(\sigma^a) - H(\sigma) \\
 &=\ -\sum_{\substack{u,v \in V \\ u\sim v}} \sigma^a_u \sigma^a_v + \sum_{\substack{u,v \in V \\ u\sim v}} \sigma_u \sigma_v \\
 &=\ -\sum_{\substack{u,v \in V \\ u\sim v}} \left( \sigma^a_u \sigma^a_v - \sigma_u \sigma_v \right)\\
 &=\ 2\sigma_a\sum_{\substack{u \in V \\ u\sim a}}\sigma_u.
\end{align*}
Thus, the transition probability from $\sigma$ to $\sigma^a$ of the Metropolis chain $P$ for the Ising model with parameter $\beta \geq 0$ is determined by the quantity
\[
\exp(-\beta\Delta H_a(\sigma)).
\]

We now turn to study a situation where each vertex $a$ has its own parameter $\beta_a$.
In other word, we shall describe a Markov chain $P_{\ind}$ that moves from $\sigma$ to $\sigma^a$ with probability depending on
\[
\exp(-\beta_a\Delta H_a(\sigma)),
\]
where $\beta_a \geq 0$ is a individual (possibly distinct) parameter for each $a \in V$.
More precisely, the probability that the new chain moves from $\sigma$ to $\sigma^a$  is defined as
\[
P_{\ind}(\sigma, \sigma^a) \ = \ \frac{1}{|V|} \left(1\ \wedge\ \exp(-\beta_a\Delta H_a(\sigma) ) \right).
\]
The probability that the chain stays at the same configuration is
\[P_{\ind}(\sigma, \sigma) \ = \ 1 - \sum_{a\in V} P_{\ind}(\sigma, \sigma^a).\]
Hence, all the configurations $\sigma'$ that differ from $\sigma$ in at least two vertices are not reachable from $\sigma$.
That is to say, $P_{\ind}(\sigma, \sigma') = 0$ if and only if $\sigma' \neq \sigma^a$ for any $a\in V$.

\begin{definition}[Ising measure with individual temperatures]
Let $G=(V,E)$  be a finite, connected graph and $(\beta_u:u\in V)$ a collection of non-negative real numbers.
The probability measure $\mu_\ind$ on $X = \{+1,-1\}^V$  is defined by
\[
\mu_\ind(\sigma) \ =\ \frac{1}{Z_\ind}
\exp\Bigg( \sum_{\substack{u,v \in V \\ u\sim v}}\beta_u \sigma_u\sigma_v\Bigg)
\]
where $Z_\ind = \sum_{\sigma \in X}\mu_\ind(\sigma)$ is a  normalizing constant.
\end{definition}

\begin{remark}
We can think of $\mu_\ind$ as an \emph{heterogenous Ising model} as opposed to the homogeneous version $\mu$ defined in Section \ref{sec:model} by
\[
\mu(\sigma)\ =\ \frac{1}{Z}  \exp\Bigg(\beta \sum_{\substack{u,v \in V \\ u\sim v}} \sigma_u\sigma_v\Bigg).
\]
\end{remark}

\begin{remark}
It is cleat that the probability measure $\mu$ is a stationary measure of the Markov chain defined by the transition matrix $P_\ind$ just in case we have $\beta_a = \beta$ for all $a \in V$.
In other words, $\mu_\ind = \mu$ if and only if the individual parameters $\beta_a$ in the definition of $P_\ind$ are all equal to the single parameter $\beta$ of the homogeneous Ising model.
\end{remark}

\begin{proposition}
The probability measure $\mu_\ind$ is the stationary measure of the Markov chain defined by the transition matrix $P_\ind$.
\end{proposition}
\begin{proof}
In order to satisfy the detailed balanced equations
\[
\mu_\ind(\sigma) P_{\ind}(\sigma, \sigma^a)  = \mu_\ind(\sigma^a)P_{\ind}(\sigma^a, \sigma)
\]
 we must have
\[
\mu_\ind(\sigma)\left(1\ \wedge\ \exp(-\beta_a\Delta H_a(\sigma) ) \right) \ =\ \mu_\ind(\sigma^a) \left(1\ \wedge\ \exp(\beta_a\Delta H_a(\sigma) ) \right)
\]
for all $\sigma$ and $\sigma^a$, because
\[\Delta H_a(\sigma^a) \ = \ H(\sigma) - H(\sigma^a) \ =\ -\Delta H_a(\sigma) .\]
Now, if $\Delta H_a(\sigma)\geq 0$ then $\beta_a\Delta H_a(\sigma) \geq 0$, and hence $\exp(\beta_a\Delta H_a(\sigma) ) \geq 1$, so
\[
\mu_\ind(\sigma) \exp(-\beta_a\Delta H_a(\sigma) )   \ =\ \mu_\ind(\sigma^a).
\]
Otherwise, if $\Delta H_a(\sigma)< 0$ then $-\beta_a\Delta H_a(\sigma) \geq 0$, and so $\exp(-\beta_a\Delta H_a(\sigma) ) \geq 1$, hence
\[
\mu_\ind(\sigma)    \ =\ \mu_\ind(\sigma^a)\exp(\beta_a\Delta H_a(\sigma) ).
\]
In both cases, we arrive at the conclusion that in order for $\mu_\ind$ to be the stationary measure of the chain defined by $P_\ind$, we must have
\begin{equation}\label{eq:2}
\frac{\mu_\ind(\sigma)}{\mu_\ind(\sigma^a)}   \ =\ \exp(\beta_a\Delta H_a(\sigma) )
\end{equation}
for every $\sigma \in X$ and  $a \in V$.

Now we proceed to prove that equation~\eqref{eq:2} holds.
After cancellation of $1 / Z_\ind$ and using properties of the exponential function, it suffices to check
\[
\sum_{\substack{u,v \in V \\ u\sim v}}\beta_u \sigma_u\sigma_v - \sum_{\substack{u,v \in V \\ u\sim v}} \beta_u \sigma^a_u\sigma^a_v \ = \ \beta_a\Delta H_a(\sigma)
\]
By inspection,
\begin{align*}
\sum_{\substack{u,v \in V \\ u\sim v}}\beta_u \sigma_u\sigma_v - \sum_{\substack{u,v \in V \\ u\sim v}} \beta_u \sigma^a_u\sigma^a_v
&= \ \sum_{\substack{u,v \in V \\ u\sim v}} (\beta_u \sigma_u\sigma_v - \beta_u \sigma^a_u\sigma^a_v) \\
&= \  2\beta_a \sigma_a \sum_{\substack{v\in V \\ a\sim v}} \sigma_v\\
&= \ \beta_a\Delta H_a(\sigma).
\end{align*}
Therefore, the probability measure $\mu_\ind$ and the transition matrix $P_\ind$ satisfy the detailed balance equations  and the result follows.
\end{proof}



\subsection{Random Boolean networks}
\label{sec:RBN}



\subsubsection{Homogeneous random Boolean networks}\label{ssub:hom_rbn}

Let $D= (V,A)$ be a directed graph.
We identify the nodes of $D$ with the genes in a gene regulatory network.
Suppose $D$ is a $k$-in regular directed graph.
Figure~\ref{fig:RBN_structure} is an example of a $2$-in regular digraph with $7$ nodes, i.e. $N=7, K=2$.

A family $(f_a)_{a\in V}$ of functions $f_a\colon \{0,1\}^k\longrightarrow\{0,1\}$ is called a \emph{Boolean network on} $D$.
Figure~\ref{fig:RBN_table} is an example of a Boolean network on a graph with $7$ nodes, and with the parameter of ``connectivity'' $k$ equal to $2$.
A Boolean network is called \emph{random} if the assignment $a\mapsto f_a$ is made at random by sampling independently and uniformly from the set of all the $2^{2^k}$ Boolean functions with $k$ inputs.
A function $\sigma\colon V \longrightarrow \{0,1\},\ a\mapsto \sigma_a$, is called a \emph{state} of the \emph{random Boolean network} on $D$.  The value $\sigma_a$ is called the \emph{state} of $a$.
The \emph{updating function} $F(\sigma)$ of a state $\sigma$ is the function $F(\sigma)\colon V\longrightarrow \{0,1\}, \ a\mapsto \sigma'_a,$ defined as
\[
\sigma'_a =  f_a(\sigma_{a_1},\dots,\sigma_{a_k}).
\]
For every $\sigma$,  we have a sequence of states $\sigma,\sigma',\sigma'',\dots$  such that each state is the updating function of the previous state in the sequence: $\sigma' = F(\sigma),\ \sigma'' = F(\sigma')$, and so on.
The sequence of states $\sigma_a,\sigma'_a,\sigma''_a,\dots$ is called the \emph{time series} of $a$.

\subsubsection{Heterogeneous random Boolean networks}
The description given in \ref{ssub:hom_rbn} corresponds to the case where the structure  and the updating scheme of the random Boolean network are  homogeneous.
Here we describe the two versions of heterogeneous random Boolean networks that were used in the simulations.
The first of these heterogeneous descriptions is structural, while the second gives rise to some sort of asynchronous dynamics.

The definition of Boolean network above makes the assumption  that every node in  the  directed graph has the same in-degree.
Now we consider Boolean networks over arbitrary (not necessarily $k$-in regular, directed) graphs.
A \emph{generalized Boolean network} on a directed graph $D$ consists of a family
$(f_a)_{a\in V}$ of functions $f_a\colon \{0,1\}^{k^-_a}\longrightarrow\{0,1\}$ with
 $k^-_a\geq 1$  the in-degree $a$.
Thus a \emph{heterogeneous random Boolean network} is a generalized Boolean network chosen uniformly at random. 

For talking about temporal heterogeneity we need to introduce asynchronous updating schemes \citep{Gershenson2002e}.
The \emph{heterogeneous updating function} of a state $\sigma$ of a random heterogeneous Boolean network on $D$ is the function $\tilde{F}(\sigma)\colon V\times \mathbb{N}\longrightarrow \{0,1\}$, defined by \[
(a,t)\mapsto
\begin{cases}
\sigma_a'&\textnormal{if}\ t\ \textnormal{is a multiple of}\ k_{a}^+ \\
\sigma_a& \textnormal{otherwise}
\end{cases}
\]
where $t$ is called the discrete \emph{time-step}, and $k_{a}^+$ is the out-degree of $a$: there are nodes $a_1,\dots,a_{k_{a}^+}$ all distinct, such that $aa_j\in E$ for $j=1,\dots,k_{a}^+$.


\subsection{Linear functions}
\label{sec:linear}

Here we observe that for linear functions, there is no difference between homogeneity and heterogeneity.
Indeed a function $f\colon \mathbb{R}^d \longrightarrow \mathbb{R}$ with $d\geq 1,$ is called \emph{linear} if for all $x,y\in\mathbb{R}^d$ and all $a,b\in\mathbb{R}$, we have \[f(ax + by)=af(x)+bf(y).\]
For $x_1,\dots,x_n\in \mathbb{R}^d,~n\geq 1$, it can be shown,
 by induction on the number of points $n$, that
\[
    f \left( \frac{1}{n} \sum_{i=1}^{n} x_i \right)
    =
    \frac{1}{n}\sum_{i=1}^{n} f(x_i).
\]
Thus, in the context of linear functions,  average value (heterogeneity) is the same as value of the average (homogeneity).










\subsubsection*{Acknowledgments}
O. Z. acknowledges support from CONACyT-SNI (Grant No. 620178). G.I. acknowledges support from AFOSR (Grant No. FA8655-20-1-7020), project EU H2020 Humane AI-net (Grant No. 952026), and CHIST-ERA project SAI (Grant No. FWF I 5205-N). C.G. acknowledges support from UNAM-PAPIIT (IN107919, IV100120, IN105122) and from the PASPA program from UNAM-DGAPA.  

\subsubsection*{Author contributions}
All authors conceived and designed the study. F.S.P., O.Z., and O.K.P. performed numerical simulations and derived mathematical results. All authors wrote the paper.

\subsubsection*{Competing interest statement}
All authors declare no competing interest.


\bibliographystyle{cgg}
\bibliography{refs,carlos}

\begin{thebibliography}{}

\bibitem[\protect\citeauthoryear{Adami}{Adami}{1995}]{Adami1995}
{\sc Adami, C.} (1995).
\newblock Self-organized criticality in living systems.
\newblock {\em Phys. Lett. A\/}~{\bf 203}: 29--32.
\newblock URL \url{http://arxiv.org/abs/adap-org/9401001}.

\bibitem[\protect\citeauthoryear{Albert and Barabasi}{Albert and
  Barabasi}{2002}]{albert2002statistical}
{\sc Albert, R.} {\sc and} {\sc Barabasi, A.-L.} (2002).
\newblock Statistical mechanics of complex networks.
\newblock {\em Reviews of Modern Physics\/}~{\bf 74}: 47--97.

\bibitem[\protect\citeauthoryear{Aldana}{Aldana}{2003}]{Aldana2003}
{\sc Aldana, M.} (2003).
\newblock Boolean dynamics of networks with scale-free topology.
\newblock {\em Physica D\/}~{\bf 185\/}~(1): 45--66.
\newblock URL \url{http://dx.doi.org/10.1016/S0167-2789(03)00174-X}.

\bibitem[\protect\citeauthoryear{Anderson}{Anderson}{1972}]{Anderson1972}
{\sc Anderson, P.~W.} (1972).
\newblock More is different.
\newblock {\em Science\/}~{\bf 177}: 393--396.

\bibitem[\protect\citeauthoryear{Ashby}{Ashby}{1956}]{Ashby1956}
{\sc Ashby, W.~R.} (1956).
\newblock {\em An Introduction to Cybernetics}.
\newblock Chapman \& Hall, London.
\newblock URL \url{http://pcp.vub.ac.be/ASHBBOOK.html}.

\bibitem[\protect\citeauthoryear{Bak, Tang, and Wiesenfeld}{Bak
  et~al\mbox{.}}{1987}]{BTW1987}
{\sc Bak, P.}, {\sc Tang, C.}, {\sc and} {\sc Wiesenfeld, K.} (1987).
\newblock Self-organized criticality: An explanation of the 1/f noise.
\newblock {\em Phys. Rev. Lett.\/}~{\bf 59\/}~(4) (July): 381--384.
\newblock URL \url{http://dx.doi.org/10.1103/PhysRevLett.59.381}.

\bibitem[\protect\citeauthoryear{Balleza, Alvarez-Buylla, Chaos, Kauffman,
  Shmulevich, and Aldana}{Balleza et~al\mbox{.}}{2008}]{Balleza:2008}
{\sc Balleza, E.}, {\sc Alvarez-Buylla, E.~R.}, {\sc Chaos, A.}, {\sc Kauffman,
  S.}, {\sc Shmulevich, I.}, {\sc and} {\sc Aldana, M.} (2008).
\newblock Critical dynamics in genetic regulatory networks: Examples from four
  kingdoms.
\newblock {\em PLoS ONE\/}~{\bf 3\/}~(6) (06): e2456.
\newblock URL \url{http://dx.plos.org/10.1371%2Fjournal.pone.0002456}.

\bibitem[\protect\citeauthoryear{Barab{\'a}si}{Barab{\'a}si}{2005}]{Barabasi2005}
{\sc Barab{\'a}si, A.-L.} (2005).
\newblock The origin of bursts and heavy tails in human dynamics.
\newblock {\em Nature\/}~{\bf 435\/}~(7039): 207--211.
\newblock URL \url{https://doi.org/10.1038/nature03459}.

\bibitem[\protect\citeauthoryear{Barab{\'a}si}{Barab{\'a}si}{2016}]{Barabasi2016}
{\sc Barab{\'a}si, A.-L.} (2016).
\newblock {\em Network Science}.
\newblock Cambridge University Press, Cambridge, UK.
\newblock URL \url{http://barabasi.com/networksciencebook/}.

\bibitem[\protect\citeauthoryear{Beggs}{Beggs}{2008}]{doi:10.1098/rsta.2007.2092}
{\sc Beggs, J.~M.} (2008).
\newblock The criticality hypothesis: how local cortical networks might
  optimize information processing.
\newblock {\em Philosophical Transactions of the Royal Society A: Mathematical,
  Physical and Engineering Sciences\/}~{\bf 366\/}~(1864): 329--343.
\newblock URL
  \url{https://royalsocietypublishing.org/doi/abs/10.1098/rsta.2007.2092}.

\bibitem[\protect\citeauthoryear{Carre{\'o}n, Gershenson, and
  Pineda}{Carre{\'o}n et~al\mbox{.}}{2017}]{10.1371/journal.pone.0190100}
{\sc Carre{\'o}n, G.}, {\sc Gershenson, C.}, {\sc and} {\sc Pineda, L.~A.}
  (2017).
\newblock Improving public transportation systems with self-organization: A
  headway-based model and regulation of passenger alighting and boarding.
\newblock {\em PLOS ONE\/}~{\bf 12\/}~(12) (12): 1--20.
\newblock URL \url{https://doi.org/10.1371/journal.pone.0190100}.

\bibitem[\protect\citeauthoryear{Chen, Quek, Chung, Saw, and Chew}{Chen
  et~al\mbox{.}}{2021}]{Chen2021}
{\sc Chen, T.}, {\sc Quek, W.~L.}, {\sc Chung, N.~N.}, {\sc Saw, V.-L.}, {\sc
  and} {\sc Chew, L.~Y.} (2021).
\newblock Analysis and simulation of intervention strategies against bus
  bunching by means of an empirical agent-based model.
\newblock {\em Complexity\/}~{\bf 2021}: 2606191.
\newblock URL \url{https://doi.org/10.1155/2021/2606191}.

\bibitem[\protect\citeauthoryear{Chialvo}{Chialvo}{2010}]{Chialvo2010}
{\sc Chialvo, D.~R.} (2010).
\newblock Emergent complex neural dynamics.
\newblock {\em Nature Physics\/}~{\bf 6\/}~(10): 744--750.
\newblock URL \url{https://doi.org/10.1038/nphys1803}.

\bibitem[\protect\citeauthoryear{Chowdhury, Santen, and
  Schadschneider}{Chowdhury et~al\mbox{.}}{2000}]{ChowdhuryEtAl2000}
{\sc Chowdhury, D.}, {\sc Santen, L.}, {\sc and} {\sc Schadschneider, A.}
  (2000).
\newblock Statistical physics of vehicular traffic and some related systems.
\newblock {\em Physics Reports\/}~{\bf 329\/}~(4-6): 199 -- 329.
\newblock URL \url{http://dx.doi.org/10.1016/S0370-1573(99)00117-9}.

\bibitem[\protect\citeauthoryear{Christensen and Moloney}{Christensen and
  Moloney}{2005}]{christensen2005complexity}
{\sc Christensen, K.} {\sc and} {\sc Moloney, N.~R.} (2005).
\newblock {\em Complexity and criticality}.
\newblock World Scientific, Singapore.

\bibitem[\protect\citeauthoryear{Cocho, Flores, Gershenson, Pineda, and
  S{\'a}nchez}{Cocho et~al\mbox{.}}{2015}]{Cocho2015}
{\sc Cocho, G.}, {\sc Flores, J.}, {\sc Gershenson, C.}, {\sc Pineda, C.}, {\sc
  and} {\sc S{\'a}nchez, S.} (2015).
\newblock Rank diversity of languages: Generic behavior in computational
  linguistics.
\newblock {\em PLoS ONE\/}~{\bf 10\/}~(4) (04): e0121898.
\newblock URL \url{http://dx.doi.org/10.1371%2Fjournal.pone.0121898}.

\bibitem[\protect\citeauthoryear{Dannemann, Boyer, and Miramontes}{Dannemann
  et~al\mbox{.}}{2018}]{Dannemann3794}
{\sc Dannemann, T.}, {\sc Boyer, D.}, {\sc and} {\sc Miramontes, O.} (2018).
\newblock L{\'e}vy flight movements prevent extinctions and maximize population
  abundances in fragile lotka{\textendash}volterra systems.
\newblock {\em Proceedings of the National Academy of Sciences\/}~{\bf
  115\/}~(15): 3794--3799.
\newblock URL \url{https://www.pnas.org/content/115/15/3794}.

\bibitem[\protect\citeauthoryear{Derrida and Pomeau}{Derrida and
  Pomeau}{1986}]{DerridaPomeau1986}
{\sc Derrida, B.} {\sc and} {\sc Pomeau, Y.} (1986).
\newblock Random networks of automata: A simple annealed approximation.
\newblock {\em Europhys. Lett.\/}~{\bf 1\/}~(2): 45--49.

\bibitem[\protect\citeauthoryear{Fern\'andez, Maldonado, and
  Gershenson}{Fern\'andez et~al\mbox{.}}{2014}]{Fernandez2013Information-Mea}
{\sc Fern\'andez, N.}, {\sc Maldonado, C.}, {\sc and} {\sc Gershenson, C.}
  (2014).
\newblock Information measures of complexity, emergence, self-organization,
  homeostasis, and autopoiesis.
\newblock In {\em Guided Self-Organization: Inception}, {M.~Prokopenko}, (Ed.).
  Emergence, Complexity and Computation, vol.~9. Springer, Berlin Heidelberg,
  19--51.
\newblock URL \url{http://arxiv.org/abs/1304.1842}.

\bibitem[\protect\citeauthoryear{Fruchart, Hanai, Littlewood, and
  Vitelli}{Fruchart et~al\mbox{.}}{2021}]{Fruchart2021}
{\sc Fruchart, M.}, {\sc Hanai, R.}, {\sc Littlewood, P.~B.}, {\sc and} {\sc
  Vitelli, V.} (2021).
\newblock Non-reciprocal phase transitions.
\newblock {\em Nature\/}~{\bf 592\/}~(7854): 363--369.
\newblock URL \url{https://doi.org/10.1038/s41586-021-03375-9}.

\bibitem[\protect\citeauthoryear{Gershenson}{Gershenson}{2002}]{Gershenson2002e}
{\sc Gershenson, C.} (2002).
\newblock Classification of random {Boolean} networks.
\newblock In {\em Artificial Life {VIII}: Proceedings of the Eight
  International Conference on Artificial Life}, {R.~K. Standish}, {M.~A.
  Bedau}, {and} {H.~A. Abbass}, (Eds.). MIT Press, Cambridge, MA, USA,
  pp.~1--8.
\newblock URL \url{http://arxiv.org/abs/cs/0208001}.

\bibitem[\protect\citeauthoryear{Gershenson}{Gershenson}{2011}]{Gershenson:2011a}
{\sc Gershenson, C.} (2011).
\newblock Self-organization leads to supraoptimal performance in public
  transportation systems.
\newblock {\em {PLoS ONE}\/}~{\bf 6\/}~(6): e21469.
\newblock URL \url{http://dx.doi.org/10.1371/journal.pone.0021469}.

\bibitem[\protect\citeauthoryear{Gershenson}{Gershenson}{2012}]{Gershenson:2010}
{\sc Gershenson, C.} (2012).
\newblock Guiding the self-organization of random {Boolean} networks.
\newblock {\em Theory in Biosciences\/}~{\bf 131\/}~(3) (September): 181--191.
\newblock URL \url{http://arxiv.org/abs/1005.5733}.

\bibitem[\protect\citeauthoryear{Gershenson}{Gershenson}{2015}]{Gershenson2015Requisite-Varie}
{\sc Gershenson, C.} (2015).
\newblock Requisite variety, autopoiesis, and self-organization.
\newblock {\em Kybernetes\/}~{\bf 44\/}~(6--7): 866--873.

\bibitem[\protect\citeauthoryear{Gershenson and Helbing}{Gershenson and
  Helbing}{2015}]{GershensonHelbing2015}
{\sc Gershenson, C.} {\sc and} {\sc Helbing, D.} (2015).
\newblock When slower is faster.
\newblock {\em Complexity\/}~{\bf 21\/}~(2): 9--15.
\newblock URL \url{http://dx.doi.org/10.1002/cplx.21736}.

\bibitem[\protect\citeauthoryear{Gershenson and Pineda}{Gershenson and
  Pineda}{2009}]{GershensonPineda2009}
{\sc Gershenson, C.} {\sc and} {\sc Pineda, L.~A.} (2009).
\newblock Why does public transport not arrive on time? {The} pervasiveness of
  equal headway instability.
\newblock {\em {PLoS ONE}\/}~{\bf 4\/}~(10): e7292.
\newblock URL \url{http://dx.doi.org/10.1371/journal.pone.0007292}.

\bibitem[\protect\citeauthoryear{Glauber}{Glauber}{1963}]{Glauber1963TimeDependent-S}
{\sc Glauber, R.~J.} (1963).
\newblock Time-dependent statistics of the {Ising} model.
\newblock {\em Journal of Mathematical Physics\/}~{\bf 4\/}~(2): 294--307.
\newblock URL \url{http://dx.doi.org/10.1063/1.1703954}.

\bibitem[\protect\citeauthoryear{Gross and Sayama}{Gross and
  Sayama}{2009}]{gross2009adaptiveNets}
{\sc Gross, T.} {\sc and} {\sc Sayama, H.}, Eds. (2009).
\newblock {\em Adaptive networks: Theory, Models and Applications}.
\newblock Understanding Complex Systems. Springer, Berlin Heidelberg.
\newblock URL \url{http://dx.doi.org/10.1007/978-3-642-01284-6}.

\bibitem[\protect\citeauthoryear{Harvey and Bossomaier}{Harvey and
  Bossomaier}{1997}]{HarveyBossomaier1997}
{\sc Harvey, I.} {\sc and} {\sc Bossomaier, T.} (1997).
\newblock Time out of joint: Attractors in asynchronous random {Boolean}
  networks.
\newblock In {\em Proceedings of the Fourth European Conference on Artificial
  Life {(ECAL97)}}, {P.~Husbands} {and} {I.~Harvey}, (Eds.). MIT Press,
  pp.~67--75.
\newblock URL \url{http://tinyurl.com/yxrxbp}.

\bibitem[\protect\citeauthoryear{Helbing}{Helbing}{2001}]{helbing2001traffic}
{\sc Helbing, D.} (2001).
\newblock Traffic and related self-driven many-particle systems.
\newblock {\em Reviews of modern physics\/}~{\bf 73\/}~(4): 1067.

\bibitem[\protect\citeauthoryear{Hesse and Gross}{Hesse and
  Gross}{2014}]{hesse2014self}
{\sc Hesse, J.} {\sc and} {\sc Gross, T.} (2014).
\newblock Self-organized criticality as a fundamental property of neural
  systems.
\newblock {\em Frontiers in systems neuroscience\/}~{\bf 8}: 166.

\bibitem[\protect\citeauthoryear{Hidalgo, Grilli, Suweis, Maritan, and
  Mu{\~{n}}oz}{Hidalgo et~al\mbox{.}}{2016}]{Hidalgo_2016}
{\sc Hidalgo, J.}, {\sc Grilli, J.}, {\sc Suweis, S.}, {\sc Maritan, A.}, {\sc
  and} {\sc Mu{\~{n}}oz, M.~A.} (2016).
\newblock Cooperation, competition and the emergence of criticality in
  communities of adaptive systems.
\newblock {\em Journal of Statistical Mechanics: Theory and Experiment\/}~{\bf
  2016\/}~(3) (mar): 033203.
\newblock URL \url{https://doi.org/10.1088/1742-5468/2016/03/033203}.

\bibitem[\protect\citeauthoryear{Holme}{Holme}{2015}]{holme2015modern}
{\sc Holme, P.} (2015).
\newblock Modern temporal network theory: a colloquium.
\newblock {\em Eur. Phys. J. B\/}~{\bf 88\/}~(9): 1--30.

\bibitem[\protect\citeauthoryear{Holme and Saram\"{a}ki}{Holme and
  Saram\"{a}ki}{2012}]{Holme2012TemporalNets}
{\sc Holme, P.} {\sc and} {\sc Saram\"{a}ki, J.} (2012).
\newblock Temporal networks.
\newblock {\em Physics Reports\/}~{\bf 519\/}~(3): 97 -- 125.
\newblock URL \url{http://arxiv.org/abs/1108.1780}.

\bibitem[\protect\citeauthoryear{Hopfield}{Hopfield}{1982}]{Hopfield1982}
{\sc Hopfield, J.~J.} (1982).
\newblock Neural networks and physical systems with emergent collective
  computational abilities.
\newblock {\em Proceedings of the National Academy of Sciences\/}~{\bf
  79\/}~(8): 2554--2558.
\newblock URL \url{https://www.pnas.org/doi/abs/10.1073/pnas.79.8.2554}.

\bibitem[\protect\citeauthoryear{I{\~n}iguez, Pineda, Gershenson, and
  Barab{\'a}si}{I{\~n}iguez et~al\mbox{.}}{2022}]{Iniguez2022}
{\sc I{\~n}iguez, G.}, {\sc Pineda, C.}, {\sc Gershenson, C.}, {\sc and} {\sc
  Barab{\'a}si, A.-L.} (2022).
\newblock Dynamics of ranking.
\newblock {\em Nature Communications\/}~{\bf 13\/}~(1): 1646.
\newblock URL \url{https://doi.org/10.1038/s41467-022-29256-x}.

\bibitem[\protect\citeauthoryear{Ising}{Ising}{1925}]{ising1925beitrag}
{\sc Ising, E.} (1925).
\newblock Beitrag zur theorie des ferromagnetismus.
\newblock {\em Zeitschrift f{\"u}r Physik\/}~{\bf 31\/}~(1): 253--258.

\bibitem[\protect\citeauthoryear{Kauffman}{Kauffman}{1969}]{Kauffman1969}
{\sc Kauffman, S.~A.} (1969).
\newblock Metabolic stability and epigenesis in randomly constructed genetic
  nets.
\newblock {\em Journal of Theoretical Biology\/}~{\bf 22}: 437--467.

\bibitem[\protect\citeauthoryear{Kauffman}{Kauffman}{1993}]{Kauffman1993}
{\sc Kauffman, S.~A.} (1993).
\newblock {\em The Origins of Order}.
\newblock Oxford University Press, Oxford, UK.

\bibitem[\protect\citeauthoryear{Langton}{Langton}{1990}]{Langton1990}
{\sc Langton, C.~G.} (1990).
\newblock Computation at the edge of chaos: Phase transitions and emergent
  computation.
\newblock {\em Physica D\/}~{\bf 42}: 12--37.

\bibitem[\protect\citeauthoryear{Lloyd}{Lloyd}{2001}]{lloyd2001measures}
{\sc Lloyd, S.} (2001).
\newblock Measures of complexity: a non-exhaustive list.
\newblock Department of Mechanical Engineering, Massachusetts Institute of
  Technology.
\newblock URL \url{http://web.mit.edu/esd.83/www/notebook/Complexity.PDF}.

\bibitem[\protect\citeauthoryear{Lopez-Ruiz, Mancini, and Calbet}{Lopez-Ruiz
  et~al\mbox{.}}{1995}]{LopezRuiz:1995}
{\sc Lopez-Ruiz, R.}, {\sc Mancini, H.~L.}, {\sc and} {\sc Calbet, X.} (1995).
\newblock A statistical measure of complexity.
\newblock {\em Physics Letters A\/}~{\bf 209\/}~(5-6): 321--326.
\newblock URL \url{http://dx.doi.org/10.1016/0375-9601(95)00867-5}.

\bibitem[\protect\citeauthoryear{Luque and Sol{\'e}}{Luque and
  Sol{\'e}}{1997}]{LuqueSole1997}
{\sc Luque, B.} {\sc and} {\sc Sol{\'e}, R.~V.} (1997).
\newblock Phase transitions in random networks: Simple analytic determination
  of critical points.
\newblock {\em Physical Review E\/}~{\bf 55\/}~(1): 257--260.
\newblock URL \url{http://tinyurl.com/y8pk9y}.

\bibitem[\protect\citeauthoryear{Mart\'inez-Ar\'evalo, Rodr\'iguez-Vazquez, and
  Gershenson}{Mart\'inez-Ar\'evalo et~al\mbox{.}}{2020}]{Martinez-Arevalo2020}
{\sc Mart\'inez-Ar\'evalo, Y.~I.}, {\sc Rodr\'iguez-Vazquez, K.}, {\sc and}
  {\sc Gershenson, C.} (2020).
\newblock Temporal heterogeneity improves speed and convergence in genetic
  algorithms.
\newblock arXiv:2203.13194.

\bibitem[\protect\citeauthoryear{McShane}{McShane}{1937}]{mcshane1937jensen}
{\sc McShane, E.~J.} (1937).
\newblock Jensen's inequality.
\newblock {\em Bulletin of the American Mathematical Society\/}~{\bf 43\/}~(8):
  521--527.

\bibitem[\protect\citeauthoryear{Molnar, Nishikawa, and Motter}{Molnar
  et~al\mbox{.}}{2021}]{Molnar2021}
{\sc Molnar, F.}, {\sc Nishikawa, T.}, {\sc and} {\sc Motter, A.~E.} (2021).
\newblock Asymmetry underlies stability in power grids.
\newblock {\em Nature Communications\/}~{\bf 12\/}~(1): 1457.
\newblock URL \url{https://doi.org/10.1038/s41467-021-21290-5}.

\bibitem[\protect\citeauthoryear{Mora and Bialek}{Mora and
  Bialek}{2011}]{Mora2011}
{\sc Mora, T.} {\sc and} {\sc Bialek, W.} (2011).
\newblock Are biological systems poised at criticality?
\newblock {\em Journal of Statistical Physics\/}~{\bf 144\/}~(2): 268--302.
\newblock URL \url{https://doi.org/10.1007/s10955-011-0229-4}.

\bibitem[\protect\citeauthoryear{Morales, Colman, S{\'a}nchez,
  S{\'a}nchez-Puig, Pineda, I{\~n}iguez, Cocho, Flores, and Gershenson}{Morales
  et~al\mbox{.}}{2018}]{10.3389/fphy.2018.00045}
{\sc Morales, J.~A.}, {\sc Colman, E.}, {\sc S{\'a}nchez, S.}, {\sc
  S{\'a}nchez-Puig, F.}, {\sc Pineda, C.}, {\sc I{\~n}iguez, G.}, {\sc Cocho,
  G.}, {\sc Flores, J.}, {\sc and} {\sc Gershenson, C.} (2018).
\newblock Rank dynamics of word usage at multiple scales.
\newblock {\em Frontiers in Physics\/}~{\bf 6}: 45.
\newblock URL
  \url{https://www.frontiersin.org/article/10.3389/fphy.2018.00045}.

\bibitem[\protect\citeauthoryear{Morales, S{\'a}nchez, Flores, Pineda,
  Gershenson, Cocho, Zizumbo, Rodr{\'\i}guez, and I{\~{n}}iguez}{Morales
  et~al\mbox{.}}{2016}]{Morales2016}
{\sc Morales, J.~A.}, {\sc S{\'a}nchez, S.}, {\sc Flores, J.}, {\sc Pineda,
  C.}, {\sc Gershenson, C.}, {\sc Cocho, G.}, {\sc Zizumbo, J.}, {\sc
  Rodr{\'\i}guez, R.~F.}, {\sc and} {\sc I{\~{n}}iguez, G.} (2016).
\newblock Generic temporal features of performance rankings in sports and
  games.
\newblock {\em EPJ Data Science\/}~{\bf 5\/}~(1): 33.
\newblock URL \url{http://dx.doi.org/10.1140/epjds/s13688-016-0096-y}.

\bibitem[\protect\citeauthoryear{Mu\~noz}{Mu\~noz}{2018}]{RevModPhys.90.031001}
{\sc Mu\~noz, M.~A.} (2018).
\newblock Colloquium: Criticality and dynamical scaling in living systems.
\newblock {\em Rev. Mod. Phys.\/}~{\bf 90}: 031001.
\newblock URL \url{https://link.aps.org/doi/10.1103/RevModPhys.90.031001}.

\bibitem[\protect\citeauthoryear{Newman}{Newman}{2003}]{Newman:2003}
{\sc Newman, M. E.~J.} (2003).
\newblock The structure and function of complex networks.
\newblock {\em SIAM Review\/}~{\bf 45}: 167--256.
\newblock URL \url{http://arxiv.org/abs/cond-mat/0303516}.

\bibitem[\protect\citeauthoryear{Nicolaou, Case, Wee, Driscoll, and
  Motter}{Nicolaou et~al\mbox{.}}{2021}]{Nicolaou2021}
{\sc Nicolaou, Z.~G.}, {\sc Case, D.~J.}, {\sc Wee, E. B. v.~d.}, {\sc
  Driscoll, M.~M.}, {\sc and} {\sc Motter, A.~E.} (2021).
\newblock Heterogeneity-stabilized homogeneous states in driven media.
\newblock {\em Nature Communications\/}~{\bf 12\/}~(1): 4486.
\newblock URL \url{https://doi.org/10.1038/s41467-021-24459-0}.

\bibitem[\protect\citeauthoryear{Nishikawa and Motter}{Nishikawa and
  Motter}{2016}]{PhysRevLett.117.114101}
{\sc Nishikawa, T.} {\sc and} {\sc Motter, A.~E.} (2016).
\newblock Symmetric states requiring system asymmetry.
\newblock {\em Phys. Rev. Lett.\/}~{\bf 117}: 114101.
\newblock URL \url{https://link.aps.org/doi/10.1103/PhysRevLett.117.114101}.

\bibitem[\protect\citeauthoryear{Oosawa and Savageau}{Oosawa and
  Savageau}{2002}]{OosawaSavageau2002}
{\sc Oosawa, C.} {\sc and} {\sc Savageau, M.~A.} (2002).
\newblock Effects of alternative connectivity on behavior of randomly
  constructed {Boolean} networks.
\newblock {\em Physica D\/}~{\bf 170}: 143--161.

\bibitem[\protect\citeauthoryear{Pineda, Kim, and Gershenson}{Pineda
  et~al\mbox{.}}{2019}]{Pineda2019}
{\sc Pineda, O.~K.}, {\sc Kim, H.}, {\sc and} {\sc Gershenson, C.} (2019).
\newblock A novel antifragility measure based on satisfaction and its
  application to random and biological {Boolean} networks.
\newblock {\em Complexity\/}~{\bf 2019}: 10.
\newblock URL \url{https://doi.org/10.1155/2019/3728621}.

\bibitem[\protect\citeauthoryear{Prokopenko, Lizier, Obst, and Wang}{Prokopenko
  et~al\mbox{.}}{2011}]{Prokopenko2011Relating-Fisher}
{\sc Prokopenko, M.}, {\sc Lizier, J.~T.}, {\sc Obst, O.}, {\sc and} {\sc Wang,
  X.~R.} (2011).
\newblock Relating {Fisher} information to order parameters.
\newblock {\em Phys. Rev. E\/}~{\bf 84}: 041116.
\newblock URL \url{http://dx.doi.org/10.1103/PhysRevE.84.041116}.

\bibitem[\protect\citeauthoryear{Ramos-Fern{\'a}ndez, Mateos, Miramontes,
  Cocho, Larralde, and Ayala-Orozco}{Ramos-Fern{\'a}ndez
  et~al\mbox{.}}{2004}]{Ramos-Fernandez2004}
{\sc Ramos-Fern{\'a}ndez, G.}, {\sc Mateos, J.}, {\sc Miramontes, O.}, {\sc
  Cocho, G.}, {\sc Larralde, H.}, {\sc and} {\sc Ayala-Orozco, B.} (2004).
\newblock L{\'e}vy walk patterns in the foraging movements of spider monkeys
  (\emph{Ateles geoffroyi}).
\newblock {\em Behavioral Ecology and Sociobiology\/}~{\bf 55\/}~(3): 223--230.
\newblock URL \url{https://doi.org/10.1007/s00265-003-0700-6}.

\bibitem[\protect\citeauthoryear{Ratnayake, Weragoda, Wansapura,
  Kasthurirathna, and Piraveenan}{Ratnayake et~al\mbox{.}}{2021}]{math9212769}
{\sc Ratnayake, P.}, {\sc Weragoda, S.}, {\sc Wansapura, J.}, {\sc
  Kasthurirathna, D.}, {\sc and} {\sc Piraveenan, M.} (2021).
\newblock Quantifying the robustness of complex networks with heterogeneous
  nodes.
\newblock {\em Mathematics\/}~{\bf 9\/}~(21).
\newblock URL \url{https://www.mdpi.com/2227-7390/9/21/2769}.

\bibitem[\protect\citeauthoryear{Roli, Villani, Filisetti, and Serra}{Roli
  et~al\mbox{.}}{2018}]{Roli2018}
{\sc Roli, A.}, {\sc Villani, M.}, {\sc Filisetti, A.}, {\sc and} {\sc Serra,
  R.} (2018).
\newblock Dynamical criticality: Overview and open questions.
\newblock {\em Journal of Systems Science and Complexity\/}~{\bf 31\/}~(3):
  647--663.
\newblock URL \url{https://doi.org/10.1007/s11424-017-6117-5}.

\bibitem[\protect\citeauthoryear{Santamar{\'\i}a-Bonfil, Fern{\'a}ndez, and
  Gershenson}{Santamar{\'\i}a-Bonfil et~al\mbox{.}}{2016}]{CxContinuous2016}
{\sc Santamar{\'\i}a-Bonfil, G.}, {\sc Fern{\'a}ndez, N.}, {\sc and} {\sc
  Gershenson, C.} (2016).
\newblock Measuring the complexity of continuous distributions.
\newblock {\em Entropy\/}~{\bf 18\/}~(3): 72.
\newblock URL \url{http://www.mdpi.com/1099-4300/18/3/72}.

\bibitem[\protect\citeauthoryear{Santamar\'ia-Bonfil, Gershenson, and
  Fern\'{a}ndez}{Santamar\'ia-Bonfil
  et~al\mbox{.}}{2017}]{10.3389/frobt.2017.00010}
{\sc Santamar\'ia-Bonfil, G.}, {\sc Gershenson, C.}, {\sc and} {\sc
  Fern\'{a}ndez, N.} (2017).
\newblock A package for measuring emergence, self-organization, and complexity
  based on {Shannon} entropy.
\newblock {\em Frontiers in Robotics and AI\/}~{\bf 4}: 10.
\newblock URL
  \url{http://journal.frontiersin.org/article/10.3389/frobt.2017.00010}.

\bibitem[\protect\citeauthoryear{Santos, Pacheco, and Lenaerts}{Santos
  et~al\mbox{.}}{2006}]{Santos2006}
{\sc Santos, F.~C.}, {\sc Pacheco, J.~M.}, {\sc and} {\sc Lenaerts, T.} (2006).
\newblock Evolutionary dynamics of social dilemmas in structured heterogeneous
  populations.
\newblock {\em Proc. Natl. Acad. Sci. USA\/}~{\bf 103}: 3490--3494.
\newblock URL \url{http://www.pnas.org/content/103/9/3490.long}.

\bibitem[\protect\citeauthoryear{Santos, Santos, and Pacheco}{Santos
  et~al\mbox{.}}{2008}]{Santos2008}
{\sc Santos, F.~C.}, {\sc Santos, M.~D.}, {\sc and} {\sc Pacheco, J.~M.}
  (2008).
\newblock Social diversity promotes the emergence of cooperation in public
  goods games.
\newblock {\em Nature\/}~{\bf 454\/}~(7201): 213--216.
\newblock URL
  \url{http://www.nature.com/nature/journal/v454/n7201/full/nature06940.html}.

\bibitem[\protect\citeauthoryear{Shannon}{Shannon}{1948}]{Shannon1948}
{\sc Shannon, C.~E.} (1948).
\newblock A mathematical theory of communication.
\newblock {\em Bell System Technical Journal\/}~{\bf 27\/}~(3 and 4) (July and
  October): 379--423 and 623--656.
\newblock URL \url{http://dx.doi.org/10.1002/j.1538-7305.1948.tb01338.x}.

\bibitem[\protect\citeauthoryear{Shmulevich, Kauffman, and Aldana}{Shmulevich
  et~al\mbox{.}}{2005}]{doi:10.1073/pnas.0506771102}
{\sc Shmulevich, I.}, {\sc Kauffman, S.~A.}, {\sc and} {\sc Aldana, M.} (2005).
\newblock Eukaryotic cells are dynamically ordered or critical but not chaotic.
\newblock {\em Proceedings of the National Academy of Sciences\/}~{\bf
  102\/}~(38): 13439--13444.
\newblock URL \url{https://www.pnas.org/doi/abs/10.1073/pnas.0506771102}.

\bibitem[\protect\citeauthoryear{Sormunen, Gross, and Saram{\"a}ki}{Sormunen
  et~al\mbox{.}}{2022}]{Sormunen2022}
{\sc Sormunen, S.}, {\sc Gross, T.}, {\sc and} {\sc Saram{\"a}ki, J.} (2022).
\newblock Critical drift in a neuro-inspired adaptive network.
\newblock arXiv:2206.10315v1.

\bibitem[\protect\citeauthoryear{Stanley}{Stanley}{1987}]{stanley1987introduction}
{\sc Stanley, H.~E.} (1987).
\newblock {\em Introduction to phase transitions and critical phenomena}.
\newblock Oxford University Press, Oxford, UK.

\bibitem[\protect\citeauthoryear{Torres-Sosa, Huang, and Aldana}{Torres-Sosa
  et~al\mbox{.}}{2012}]{TorresSosa2012}
{\sc Torres-Sosa, C.}, {\sc Huang, S.}, {\sc and} {\sc Aldana, M.} (2012).
\newblock Criticality is an emergent property of genetic networks that exhibit
  evolvability.
\newblock {\em PLoS Comput Biol\/}~{\bf 8\/}~(9) (09): e1002669.
\newblock URL \url{http://dx.doi.org/10.1371%2Fjournal.pcbi.1002669}.

\bibitem[\protect\citeauthoryear{Vicsek and Zafeiris}{Vicsek and
  Zafeiris}{2012}]{Vicsek2012}
{\sc Vicsek, T.} {\sc and} {\sc Zafeiris, A.} (2012).
\newblock Collective motion.
\newblock {\em Physics Reports\/}~{\bf 517}: 71--140.
\newblock URL \url{http://dx.doi.org/10.1016/j.physrep.2012.03.004}.

\bibitem[\protect\citeauthoryear{Vidiella, Guillamon, Sardany{\'e}s, Maull,
  Conde-Pueyo, and Sol{\'e}}{Vidiella
  et~al\mbox{.}}{2020}]{Vidiella2020.11.16.385385}
{\sc Vidiella, B.}, {\sc Guillamon, A.}, {\sc Sardany{\'e}s, J.}, {\sc Maull,
  V.}, {\sc Conde-Pueyo, N.}, {\sc and} {\sc Sol{\'e}, R.} (2020).
\newblock Engineering self-organized criticality in living cells.
\newblock {\em bioRxiv\/}~{\bf 2020.11.16.385385}.
\newblock URL
  \url{https://www.biorxiv.org/content/early/2020/11/17/2020.11.16.385385}.

\bibitem[\protect\citeauthoryear{Wang, Lizier, and Prokopenko}{Wang
  et~al\mbox{.}}{2011}]{wang2011fisher}
{\sc Wang, X.}, {\sc Lizier, J.}, {\sc and} {\sc Prokopenko, M.} (2011).
\newblock Fisher information at the edge of chaos in random {Boolean} networks.
\newblock {\em Artificial Life\/}~{\bf 17\/}~(4): 315--329.
\newblock Special Issue on Complex Networks.
\newblock URL \url{http://dx.doi.org/10.1162/artl_a_00041}.

\bibitem[\protect\citeauthoryear{Wang, Lizier, and Prokopenko}{Wang
  et~al\mbox{.}}{2010}]{Wang:2010}
{\sc Wang, X.~R.}, {\sc Lizier, J.}, {\sc and} {\sc Prokopenko, M.} (2010).
\newblock A {Fisher} information study of phase transitions in random {Boolean}
  networks.
\newblock In {\em {Artificial Life XII} Proceedings of the Twelfth
  International Conference on the Synthesis and Simulation of Living Systems},
  {H.~Fellermann}, {M.~D\"{o}rr}, {M.~M. Hanczyc}, {L.~L. Laursen},
  {S.~Maurer}, {D.~Merkle}, {P.-A. Monnard}, {K.~St$\o$y}, {and}
  {S.~Rasmussen}, (Eds.). MIT Press, Odense, Denmark, 305--312.
\newblock URL \url{http://tinyurl.com/37qxgtn}.

\bibitem[\protect\citeauthoryear{Wolpert and Macready}{Wolpert and
  Macready}{1995}]{wolpert95no}
{\sc Wolpert, D.~H.} {\sc and} {\sc Macready, W.~G.} (1995).
\newblock No free lunch theorems for search.
\newblock Tech. Rep. SFI-WP-95-02-010, Santa Fe Institute.
\newblock URL \url{http://tinyurl.com/yz274ej}.

\bibitem[\protect\citeauthoryear{Wolpert and Macready}{Wolpert and
  Macready}{1997}]{wolpert1997no}
{\sc Wolpert, D.~H.} {\sc and} {\sc Macready, W.~G.} (1997).
\newblock {No Free Lunch Theorems for Optimization}.
\newblock {\em IEEE Transactions on Evolutionary Computation\/}~{\bf 1\/}~(1):
  67--82.

\bibitem[\protect\citeauthoryear{Zhang, Ocampo-Espindola, Kiss, and
  Motter}{Zhang et~al\mbox{.}}{2021}]{Zhange2024299118}
{\sc Zhang, Y.}, {\sc Ocampo-Espindola, J.~L.}, {\sc Kiss, I.~Z.}, {\sc and}
  {\sc Motter, A.~E.} (2021).
\newblock Random heterogeneity outperforms design in network synchronization.
\newblock {\em Proceedings of the National Academy of Sciences\/}~{\bf
  118\/}~(21).
\newblock URL \url{https://www.pnas.org/content/118/21/e2024299118}.

\bibitem[\protect\citeauthoryear{Zhou, Lu, and Holme}{Zhou
  et~al\mbox{.}}{2020}]{ZHOU202047}
{\sc Zhou, B.}, {\sc Lu, X.}, {\sc and} {\sc Holme, P.} (2020).
\newblock Universal evolution patterns of degree assortativity in social
  networks.
\newblock {\em Social Networks\/}~{\bf 63}: 47--55.
\newblock URL
  \url{https://www.sciencedirect.com/science/article/pii/S0378873320300253}.

\end{thebibliography}

\end{document}